\newif\iffull

\fulltrue 
\documentclass{llncs}

\iffull
\usepackage{fullpage}
\else
\fi

\usepackage[utf8]{inputenc}
\DeclareUnicodeCharacter{1D45D}{}
\usepackage{multicol}
\usepackage{appendix}
\usepackage{tikz}
\usetikzlibrary{shapes,arrows.meta, positioning}
\usepackage[export]{adjustbox}
\usepackage{mathtools}
\usepackage{caption}
\usepackage{subcaption}
\usepackage{hyperref}
\usepackage[nameinlink,capitalize,noabbrev]{cleveref}
\usepackage{url}
\usepackage{graphicx,graphics}
\usepackage{tikz}
\usetikzlibrary{arrows.meta}
\usepackage{xspace}
\usepackage{paralist}
\usepackage[bottom]{footmisc}
\usepackage[noend]{algpseudocode}
\usepackage[
        lambda,
	n,
	advantage,
	operators,
	sets,
	adversary,
	landau,
	probability,
	notions,
	logic,
	ff,
	mm,
	primitives,
	events,
	complexity,
	asymptotics,
	keys
]{cryptocode}
\usepackage{dashbox}
\usepackage[disable]{todonotes}
\usepackage[font=small,labelfont=bf,justification=centering]{caption}

\usepackage{amssymb}

\def\sfM{{\mathsf{M}}}
\def\sfR{{\mathsf{R}}}

\def\C{{\mathcal{C}}}
\def\M{{\mathcal{M}}}

\def\E{{\mathcal{E}}}
\def\L{{\mathcal{L}}}
\def\A{{\mathcal{A}}}
\def\F{{\mathcal{F}}}

\def\M{{\mathcal{M}}}

\def\false{{\mathit{false}}}
\def\true{{\mathit{true}}}

\def\sfW{{\mathsf{W}}}

\def\GC{\mathsf{GC}}

\def\GC{\mathsf{GC}}

\def\BA{\mathsf{BA}}

\def\sfA{\mathsf{A}}

\def\echo{\mathsf{echo}}

\def\output{\mathsf{output}}

\def\vote{\mathsf{vote}}
\def\echotwo{\mathsf{echo2}}
\def\voteone{\mathsf{vote1}}
\def\votetwo{\mathsf{vote2}}
\def\GCone{\mathsf{GC_{Auth}^*}}
\def\GCtwo{\mathsf{GC_{Sab}^*}}
\def\GCIT{\mathsf{\Pi_{AW}}}
\def\GCPKI{\mathsf{\Pi_{\sfM\sfR}}}
\def\Sync{\mathsf{Sync}}

\def\Piauth{\BA_{\mathsf{Auth}}}
\def\Pisab{\BA_{\mathsf{Sab}}}
\def\auth{\mathsf{Auth}}
\def\Sab{\mathsf{Sab}}

\def\CC{\mathsf{CC}}

\def\it{\mathsf{it}}

\def\part{\mathsf{part}}

\def\pk{\mathsf{pk}}
\def\sk{\mathsf{sk}}

\def\decide{\mathsf{decide}}

\usepackage{comment}

\newcommand{\name}{\textsf{Juggernaut}\xspace}

\newcommand{\daniel}[1]{{\color{red}{[Daniel: #1}]}}
\newcommand{\yuval}[1]{{\color{blue}{[Yuval: #1}]}}

\iffull
\title{\name: Efficient Crypto-Agnostic Byzantine Agreement}
\else
\title{\name: Efficient Crypto-Agnostic\\ Byzantine Agreement}
\fi

\iffull
\author{Daniel Collins\inst{1,2} \and Yuval Efron\inst{3} \and Jovan Komatovic\inst{4}\thanks{Most of this work was completed while Jovan Komatovic was at a16z crypto.}}
\institute{Purdue University \and Georgia Institute of Technology \and Columbia University \and École Polytechnique Fédérale de Lausanne (EPFL)\\
\email{colli594@purdue.edu, ye2210@columbia.edu, jovan.komatovic@epfl.ch}}
\fi

\begin{document}
\maketitle

\begin{abstract}
It is well known that a trusted setup allows one to solve the Byzantine agreement problem in the presence of $t<n/2$ corruptions, bypassing the setup-free $t<n/3$ barrier.
Alas, the overwhelming majority of protocols in the literature have the caveat that their security crucially hinges on the security of the cryptography and setup, to the point where if the cryptography is broken, even a single corrupted party can violate the security of the protocol.
Thus these protocols provide higher corruption resilience ($n/2$ instead of $n/3$) for the price of increased assumptions. 
Is this trade-off necessary?\\

We further the study of \emph{crypto-agnostic} Byzantine agreement among $n$ parties that answers this question in the negative.
Specifically, let $t_s$ and $t_i$ denote two parameters such that (1) $2t_i + t_s < n$, and (2) $t_i \leq t_s < n/2$.
Crypto-agnostic Byzantine agreement ensures agreement among honest parties if (1) the adversary is computationally bounded and corrupts up to $t_s$ parties, or (2) the adversary is computationally unbounded and corrupts up to $t_i$ parties, and is moreover given all secrets of all parties established during the setup.
We propose a compiler that transforms any pair of resilience-optimal Byzantine agreement protocols in the authenticated and information-theoretic setting into one that is crypto-agnostic.
Our compiler has several attractive qualities, including using only $O(\lambda n^2)$ bits over the two underlying Byzantine agreement protocols, and preserving round and communication complexity in the authenticated setting.
In particular, our results improve the state-of-the-art in bit complexity by at least two factors of $n$ and provide either early stopping (deterministic) or expected constant round complexity (randomized).
We therefore provide fallback security for authenticated Byzantine agreement \emph{for free} for $t_i \leq n/4$.
\end{abstract}

\section{Introduction}
Byzantine agreement (BA), the problem of reaching agreement between $n$ parties, of which at most $t$ are \emph{corrupt}, and controlled by an \emph{adversary}, is arguably the core problem in fault-tolerant distributed computing, with research spanning more than four decades~\cite{PeaseSL80,RSA:GarKia20}.
In this paper, we focus on the synchronous case, in which messages are delivered to parties at most $\Delta$ time after being sent.
Cryptographic tools and assumptions are central in the design of BA protocols, both for improved efficiency in various regimes and as well as to circumvent lower bounds~\cite{dolev1983authenticated,RSA:GarKia20,EC:GKOPZ20,TCC:BKLL20}.
Perhaps the most eminent of those cryptographic tools are digital signatures, typically instantiated alongside a public key infrastructure (PKI) assumption, in which it is assumed that on top of knowing a list of identifiers of all parties participating in the protocol, each identifier has a corresponding public-secret key pair $(pk,sk)$ with $pk$ being known to all parties.
By levearging PKI, it is well known that BA can be solved in the presence of $t<\frac{n}{2}$ corrupt parties~\cite{dolev1983authenticated,C:KatKoo06}, while setup-free protocols must assume $t<\frac{n}{3}$ (even assuming cryptography like signatures)~\cite{PeaseSL80,LSP82}.

The reliance on PKI mandates two highly crucial assumptions.
First, that any underlying cryptography remains secure.\footnote{This does not apply to protocols based on primitives like pseudosignatures~\cite{pfitzmann1996information} that are information theoretically-secure but require setup; these protocols generally have a high cost and are not deployed at present.}
Second, that the secrets established at setup remain secure.
The vast majority of literature, and all practical work on BA that assumes PKI, suffers from the following shortcoming: the security of the protocol hinges on the security of the employed cryptographic primitives, to the point where even a \emph{single} corrupt party can violate security, if the cryptography used turned out to be broken.
This precarious state of affairs is not only a theoretical concern, with perhaps the most notorious example being the transaction malleability attack in Bitcoin which resulted in losses of hundreds of millions of dollars~\cite{ESORICS:DecWat14}.
Reliance on computational assumptions is more generally risky as they may, at any time, be publicly (let alone discreetly) broken, either classically or due to the looming threat of quantum computation.
In some sense, despite the weaker corruption resilience that information-theoretic, setup-free protocols offer, they have the benefit of having no other potential weak spots in their security.

Can we get the best of both worlds? That is, a BA protocol that has optimal resilience given PKI and secure cryptography (the \emph{authenticated} setting), that still maintains high security against a computationally unbounded adversary that can nullify any setup (the \emph{sabotaged} setting)?
Note that the specific setting in which the protocol executes is chosen by the adversary at the beginning of the protocol, and in particular \emph{honest} (i.e., non-corrupt) parties are in general \emph{oblivious} to the actual setting in which the protocol executes.
Designing such protocols is precisely the question we address in this paper.

This question was in fact studied two decades ago by Fitzi, Holenstein and Wullschleger~\cite{EC:FitHolMul04} in the broader context of secure multi-party computation (MPC), in which they design an MPC (and thus a BA) protocol that has what they call \emph{hybrid} security.
In particular, it can tolerate up $t_s$ corrupt parties against a computationally bounded adversary and secure cryptography, and up to $t_i$ corrupt parties under no computational or setup assumptions, for any $t_i\leq t_s < \frac{n}{2}$ such that $t_s+2t_i<n$.\footnote{The model of \cite{EC:FitHolMul04} (and \cite{cryptoeprint:2009/434}) as stated here does not consider passive compromise of the setup as we do, and additionally considers inconsistent PKI which imposes different resilience bounds; we discuss this further later.}
They also prove that this bound is tight, even for BA.
In particular, a protocol can support up to $t_i \leq \frac{n}{4}$ faults with \emph{no loss of resilience} given a computationally-bounded adversary, i.e., with optimal $t_s < \frac{n}{2}$ fault tolerance!
While an impressive feasibility result for general MPC, when one focuses on BA, existing protocols, namely from \cite{EC:FitHolMul04} and subsequent work~\cite{cryptoeprint:2009/434}, suffer from several drawbacks hindering their usability.



\sloppy{
\paragraph{Communication Complexity.} Fitzi et al.~\cite{EC:FitHolMul04} propose a hybrid broadcast protocol with $O(\lambda n^4)$ communication complexity (in bits) that they use as a subroutine to solve MPC.
Subsequent work~\cite{cryptoeprint:2009/434} also builds broadcast with complexity $O(\lambda n^4)$ as written and using state-of-the-art sub-routines $O(\lambda n^3 \log{n})$ or $O(\lambda n^3 + n^3 \log^2{n})$ bits~\cite{cryptoeprint:2023/1172,civit2024error,EC:AshCha24} (ignoring problems with composition due to non-simultaneous termination~\cite{JC:CCGZ19}).
Building crypto-agnostic BA from parallel broadcast generically~\cite{civit2023validity} would therefore require at least $O(\lambda n^4)$ bits.
This leaves a large gap from the classic $\Omega(n^2)$ lower bound on the communication complexity of BA~\cite{dolev1985bounds}.}


\paragraph{Round Complexity.} It is well known that any deterministic BA protocol resilient against up to at most $t$ Byzantine parties must take at least $t+1$ rounds in the worst case~\cite{dolev1983authenticated}.
This can be circumvented if the protocol is \emph{early stopping}, thereby using just $O(f)$ rounds when $f < t$ corruptions actually occur~\cite{perry1984authenticated,dolev1990early}, or if it is \emph{randomized}, where expected constant-round protocols are known~\cite{C:KatKoo06,AbrahamDDN019}.
Existing hybrid protocols as written however are deterministic and are not early stopping, requiring $O(n)$ rounds of communication in all cases, and it is not clear that the protocol of \cite{cryptoeprint:2009/434} can immediately lead to constant-round BA since parallel composition of expected constant-round protocols generically results in expected $O(\log{n})$ round complexity~\cite{ben2003resilient}, let alone the high communication this would incur even if it did work.


\paragraph{General Compiler.} Existing hybrid protocols are either directly concrete \cite{EC:FitHolMul04} or make use of a linear number of instances of underlying building blocks like broadcast and are thus not amenable to efficient implementation~\cite{cryptoeprint:2009/434}.
Ideally, one would want to be able to construct an efficient protocol $\Pi$ with hybrid security in a \emph{black box} manner from given protocols $\Piauth,\Pisab$ for the authenticated and sabotaged settings, respectively, without having to solve BA from scratch.

\paragraph{Our Contributions.} Our main contribution is a compiler that enjoys all of the above properties.
Our compiler transforms any two given protocols $\Piauth,\Pisab$ in the authenticated and sabotaged settings, respectively, into a protocol $\name$ with crypto-agnostic security with optimal resilience $t_s+2t_i<n$, $t_i\leq t_s<\frac{n}{2}$.
Furthermore, $\name$ uses $\Piauth,\Pisab$ in a black-box manner, $\name$ has an additive factor of just $O(\lambda n^2)$ bits of communication over $\Piauth,\Pisab$.
Our protocol optimizes for the practical authenticated case: if $\Piauth$ is early stopping, then so is $\name$ in the authenticated setting. 
Moreover, if $\Piauth$ is a randomized protocol with expected round complexity $R$, then $\name$ has expected round complexity $O(R)$ in the authenticated setting.
Therefore, our protocol effectively provides crypto-agnostic security to an authenticated protocol \emph{for free}.

Along the way, we propose two new graded consensus gadgets with $O(\lambda n^2)$ bit complexity and constant (worst-case) round complexity that provide partial security guarantees in one world (authenticated resp. sabotaged) and full security in the other (sabotaged resp. authenticated) that may be of independent interest.

Using our compiler, we propose two concrete protocols, one deterministic and one randomized.
Our deterministic protocol has $O(\lambda n^2)$ bit complexity in all cases, has $O(f)$ round complexity for $f$ actual failures in the authenticated case and uses $O(n)$ rounds in the sabotaged case.
Our randomized protocol has $O(\lambda n^2)$ expected bit complexity and constant expected round complexity in the authenticated case, and uses $O(\lambda^2 n^2)$ bits and $O(\lambda + f)$ rounds in the sabotaged case.


\subsection{Technical Overview}\label{ssec:technical_overview}
We first would like to stress the complexity of the problem by examining state-of-the-art authenticated BA protocols achieving optimal corruption resilience.
Intuitively, without setup and given $t<\frac{n}{3}$, a \emph{quorum} consisting of at least $\frac{2}{3}$ of parties suffices to convince an honest party to adopt a value, as a counting argument shows that no quorum for a different value can exist.
This is no longer the case when one demands $t<\frac{n}{2}$, and the overwhelming majority of protocols make use of signature based \emph{equivocation checks} to assert that only one value will be adopted by honest parties during the protocol.
Any attempt to increase the size of a quorum can be met with silence from corrupt parties, resulting an unhalting executions due to $t<\frac{n}{2}$.
On the other hand, any attempt to relax equivocation checks can be met with agreement violation attacks by corrupt parties.
This forces one to rethink the problem from a first principles approach.

\paragraph{A Strawman Solution: Black Boxes and Graded Consensus.}
A natural approach is to use protocols secure in each setting as black boxes.
Let $\Piauth,\Pisab$ be protocols solving BA in the authenticated setting and sabotaged (i.e., setup-free and information-theoretic) setting, respectively.
Intuitively, we would like to run $\Piauth$ first, check if agreement was reached, and if not, run $\Pisab$.
A typical tool in the literature for detecting pre-existing agreement is the \emph{graded consensus} (GC) primitive, which allows parties to output, along with their value, a grade indicating their level of confidence in the output.
The literature luckily contains efficient implementations of GC in both the authenticated setting~\cite{Momose2021} and the sabotaged setting~\cite{AttiyaW23}.
Alas, we are faced with a trickier scenario.
Recall that the specific setting in which the protocol runs is chosen by the adversary at the beginning of protocol execution, and in particular they can choose a setting that renders any existing GC useless and provides no guarantees.

\paragraph{Building Our \name\ Protocol.} 
An observation we make, inspired by similar techniques from \emph{network-agnostic} protocols~\cite{TCC:BluKatLos19,C:BluZhaLos20} that provide security under synchrony \emph{or} asynchrony, is that we can design GCs that work as usual in one of the settings and provide partial guarantees in the other setting in order to build a crypto-agnostic protocol.
Designing such GCs with optimal corruption resilience $(t_s+2t_i<n,t_i\leq t_s<\frac{n}{2})$, $O(\lambda n^2)$ bit complexity and constant round complexity, as well as appropriately combining everything together (which brings up technical challenges, as explained below), are the main technical contributions of this paper.

Recall that Byzantine agreement provides \emph{consistency} (all parties output the same value), \emph{termination} (all parties output a value and halt), and some \emph{validity} property (in this work, namely that if all parties input the same value, that value is decided).
We consider graded consensus with \emph{two} grades, either 0 or 1: \emph{graded validity} then requires all honest parties output input $v$ and grade 1, and \emph{graded consistency} requires that if any honest party outputs $(v, 1)$, then all honest parties output $(v, 0)$ or $(v, 1)$.

As sketched above, our high level approach is to run $\Piauth$, check if agreement was reached, and run $\Pisab$ if not.
However, if we are in the sabotaged setting, $\Piauth$ can behave arbitrarily, and if we are in the authenticated setting, $\Pisab$ can behave arbitrarily.
We introduce two graded consensus protocols to deal with this.
Our first, $\GCone$, provides \emph{full} security in the authenticated case for up to $t_s$ corruptions, and ensures \emph{validity and termination} for up to $t_i$ corruptions in the sabotaged case.
Our second, $\GCtwo$, provides the \emph{opposite} guarantees: full security in the sabotaged case with $t_i$ corruptions, and validity and termination in the authenticated case with $t_s$ corruptions.

At a high level, our protocol, which we call \name, proceeds as follows.
First, each party $p_i$ runs $\GCone$ using their input $v_i$, which outputs pair $(v_1, g_1)$.
$v_1$ is then fed to $\Piauth$, which outputs $v_2$.
In the authenticated case, $\GCone$ and $\Piauth$ provide full security, and therefore all honest parties output the same $v_2$ from $\Piauth$. 
In the sabotaged case, however, we only have \emph{validity} of $\GCone$.
To preserve validity in the sabotaged case, parties then input $v_1$ to $\GCtwo$ if $g_1 = 1$, and otherwise input $v_2$ to $\GCtwo$; let the output of $\GCtwo$ be $(v_3, g_3)$.
Since $\GCtwo$ provides validity in the authenticated case, all honest parties will output the same $(v_2, 1)$, where $v_1 = v_2$ in `valid' runs of \name\ (so validity is preserved from $\GCone$ up to this point).

At this point we do not yet have consistency in the sabotaged case, only validity.
Therefore, all parties run $\Pisab$ with input $v_3$, which provides full security in the sabotaged case.
However, $\Pisab$ provides \emph{no security} in the authenticated case.
To rectify this, as well as to provide early stopping in the authenticated case, parties multicast their output value $v_3$ as $(\decide, v_3)$ \emph{only if} $g_3$ = 1 and wait for $\Delta$ time.
Then, on receipt of $n - t_s$ $(\decide, v_3)$ messages (which is always guaranteed in the authenticated case), parties output $v_3$ and can safely halt.
In the sabotaged case, however, some but not all parties may terminate at this stage.
Thus, to ensure all other parties terminate, parties terminate if they receive $(\decide, v)$ from $n - t_s - t_i$ parties \emph{after} running $\Pisab$ (if they have not yet halted).
Therefore, if an honest party halts due to receiving $n - t_s$ $(\decide, v_3)$ messages, all honest parties will receive $n - t_s - t_i$ $(\decide, v)$ messages and halt.
Otherwise, all honest parties will not halt before terminating from $\Pisab$.
In this case, if an honest party receives $n - t_s - t_i$ $(\decide, v)$ messages, then since $n - t_s - t_i > t_i + 1$, one honest party must have output $(v, 1)$ from $\GCtwo$, and by graded consistency, all honest parties output $v$, and by consistency in the sabotaged case of $\Pisab$, all honest parties will output $v$ from that, and thus agree on $v$.
Otherwise, consistency of $\Pisab$ ensures the consistency of $\Pisab$.

\paragraph{Dealing with Non-Simultaneous Termination.} The above works well if $\Piauth$ is such that all parties terminate at the same time.
However, if $\Piauth$ is early-stopping or randomized, parties will not in general output from $\Piauth$ at the same time.
For instance, the adversary can force one honest party to produce an output significantly earlier than any other honest party.
To rectify this, we utilize the \emph{synchronizer} primitive, which ensures that honest parties ``move on'' from $\Piauth$ at roughly the same time.
Concretely, the synchronizer guarantees that honest parties quit executing $\Piauth$ within at most one round of each other, regardless of whether we are in the authenticated or sabotaged setting.
Furthermore, in the authenticated setting, the synchronizer ensures that honest parties progress from $\Piauth$ with (asymptotically) no additional round overhead: (1) if $\Piauth$ is early stopping with complexity $O(f$), then honest parties progress in $O(f)$ rounds, and (2) if $\Piauth$ is randomized with expected round complexity $R$, then honest parties progress in $O(R)$ rounds.
This is essential to guarantee that, in the authenticated case, \name introduces no asymptotic round complexity overhead.



\paragraph{Constructing Crypto-Agnostic Graded Consensus.}
We provide efficient compilers that use a single instance of graded consensus protocol secure in a given setting that provide validity and termination in the other setting.
Namely, our two compilers each incur \emph{three} additional rounds and $O(\lambda n^2)$ communication overhead over the black boxes we use, for example the graded consensus protocols of Momose-Ren~\cite{Momose2021} and Attiya-Welch~\cite{AttiyaW23} which themselves have constant round complexity and quadratic communication complexity.

Our first protocol $\GCone$ provides full authenticated security and validity and termination in the sabotaged case.
Recall that the underlying authenticated graded consensus protocol, say $\GCPKI$, in general provides \emph{no} security in the sabotaged case.
Our goal is thus to augment $\GCPKI$ with a procedure to ensure sabotaged $t_i$-validity and termination.
The main challenge is ensuring that this procedure does not interfere with the authenticated security of $\GCPKI$.

The key observation lies in the fact that if the parties aren't in the validity case, then parties can change their inputs to $\GCPKI$ arbitrarily without violating security.
Therefore, parties in our protocol cast votes in search of a sufficiently large quorum (of $n - t_i$ parties) to output a value with grade $1$.
If such a quorum is found, a certificate of the quorum is made and broadcast to the rest of the parties. Upon receiving a unique certificate of this kind, a party replaces its input with the received value, and enters 
$\GCPKI$ with the new input. Validity in the sabotaged case then becomes immediate, and careful analysis is required to ensure that this added part of the protocol can not violate the $t_s$-security of the protocol in the authenticated setting.

Our second protocol $\GCtwo$ provides the opposite: sabotaged security and validity and termination in the authenticated case.
The main challenge stems from the fact that honest parties might not initiate the protocol $\GCtwo$ in the same round, due to possibly different exit times from $\Piauth$. Thanks to the Synchronizer (see \cref{ssec:synchronizer}), which leverages the synchrony assumption of the network, we know that honest parties commence $\GCtwo$ at most 1 round apart from one another.
This allows us to design $\GCtwo$ with that in mind, and not deal with the general case of asynchrony.

Similarly to $\GCone$, our augmenting of $\GCIT$ with sabotaged $t_s$-validity relies on the observation that when not in the validity case, parties may change their inputs arbitrarily without harming the security of the protocol.
A first attempt at augmenting $\GCIT$ with authenticated $t_s$-validity might look as follows: Echo the inputs, and look for a quorum of $n-t_s$ echos for a value $v$, broadcast a certificate $\C(v)$ of this quorum (if found) using threshold signatures, and if no conflicting certificates were received, cast a vote for $v$, deciding it with grade $1$ if a quorum of $n-t_s$ of votes was received.
Validity in the authenticated case clearly holds, but alas, in the sabotaged setting, for reasons that become clear in the analysis, this approach fails. 

The key observation here is that in the authenticated setting, we only care about validity, and so we can impose stricter conditions on deciding a value prior to running $\GCIT$. Specifically, our solution stipulates that witnessing a unique certificate for a value is no longer sufficient for party to decide a value, it must have also \emph{created} a certificate itself. This saves us from consistency violations in the \emph{sabotaged} setting by making sure that if an honest party decides a value before $\GCIT$, then all honest parties have seen a certificate for that value.

\subsection{Related Work}\label{ssec:related_work}
\paragraph{Hybrid Security.} Two previous works that we are aware of consider fallback security w.r.t. an unbounded adversary for an authenticated protocol~\cite{EC:FitHolMul04,cryptoeprint:2009/434} (both cited above), both focused on the feasibility of MPC.
These works additionally allow the adversary to completely compromise the PKI, given the adversary corrupts up to $t_p$ parties - they call the resulting model \emph{hybrid security}.
They show to provide security for $t_p > 0$ that $2t_s + t_p < n$ is necessary and sufficient.\footnote{In \cite{cryptoeprint:2009/434}, the authors show for $t_p = 0$ that $t_s + 2 t_i < n$ is enough, and so $t_s \leq n/2$ and $t_i > 0$ (in their model) is possible for functionalities like broadcast and MPC with (in their case unanimous) abort.}
Thus, to guarantee any security in this case, one must sacrifice resilience in $t_s$.
Our model further differs from previous work in that in the sabotaged case, we additionally allow the adversary to passively compromise the setup even under $t_s + 2t_i < n$, whereas the adversary cannot not do so in~\cite{EC:FitHolMul04,cryptoeprint:2009/434} unless $t_p$ or less corruptions are made.
This is of particular note for information-theoretic authenticated BA with fallback since one can set $t_s = \frac{n}{2} - 1$ and still achieve fallback security for $t_i \leq \frac{n}{4}$ under passively compromised setup.

As noted above, \cite{EC:FitHolMul04} build hybrid broadcast with $O(\lambda n^4)$ bit and $O(n)$ round complexity.
They first build `weak broadcast' which provides security in their hybrid setting, then generically build graded consensus with $O(n)$ instances of weak broadcast, followed by $O(n)$ instances of graded consensus (one per round) for the final broadcast protocol.
We do not see how to easily reduce this complexity without starting from scratch, which we indeed do in this work.

\cite{cryptoeprint:2009/434} also build hybrid broadcast using $O(\lambda n^4)$ bits and $O(n)$ rounds as written; we now consider its most expensive components.
First, they run Dolev-Strong broadcast w.r.t. the sender (or as we note, any $t < n$ broadcast protocol).
Then, they run so-called broadcast with extended validity, which they build from $n$ instances of perfectly-secure broadcast~\cite{berman1989towards} with a message-signature pair as input.
Finally, they run parallel broadcast where each honest party may input $O(n)$ signatures on a message.
Perfectly secure parallel broadcast with $O(\lambda)$-sized input can be built using $O(\lambda n^3 \log{n})$ bits and $O(f)$ rounds for $f$ actual corruptions~\cite{civit2024error}, and otherwise $O(\lambda n^2 + n^3 \log^2{n})$ bits and constant rounds~\cite{EC:AshCha24}.
Authenticated parallel broadcast under honest majority with $O(\lambda + n)$-sized inputs (using multisignatures) can be built using $O(n^3 + \lambda n^2)$ bits and $O(n)$ rounds~\cite{civit2024dare}, and otherwise $O(\lambda n^3)$ bits and constant rounds~\cite{cryptoeprint:2023/1172}.
Ultimately, we cannot escape using at least $O(\lambda n^3)$ bits using the approach of \cite{cryptoeprint:2009/434} to build constant-round broadcast, let alone BA with constant round complexity.

\paragraph{Fallback Security.}
Many works consider providing additional security guarantees to primitives like BA or MPC on top of or in exchange for some security in the `base' setting (in our work the authenticated setting) going back at least to Chaum~\cite{C:Chaum89} in MPC; we survey some below.
\cite{gordon2010authenticated} considers a model where the secrets of $t_c$ parties can be exposed to the adversary \emph{and} $t_a$ additional parties can be corrupted, showing in particular $2 t_a + \min\{t_a, t_c\} < n$ for (fixed) $t_a, t_c > 0$ is sufficient and necessary; observe that their model is incomparable to ours.
An \emph{accountable} BA protocol~\cite{crime_punishment} provides security given $t$ corruptions, and given $t' > t$ corruptions, parties can generate a proof that some parties must have behaved maliciously.
This resembles MPC with identifiable abort~\cite{C:IshOstZik14}, namely MPC under corrupt majority that ensures a corrupt party is unanimously detected if the protocol aborts, but is not confined to the synchronous setting.
\cite{AC:LLMMT20} considers synchronous MPC under $t_s$ corruptions with \emph{responsiveness} (like asynchronous protocols) under $t_r$ corruptions, and achieve a comparable bound as us, namely $t_s + 2t_r < n$.
\cite{lucas2010hybrid} considers trade-offs between information-theoretic robustness (preventing adversarial abortion) and computational privacy assuming broadcast and secure channels.

A line of work initiated in~\cite{TCC:BluKatLos19} considers \emph{network-agnostic} security, that is, providing security for up to $t_a$ corruptions if the network is asynchronous and $t_s$ if it is synchronous.
Some works consider feasibility results, including \cite{TCC:BluKatLos19} for BA, which is possible if and only if $2t_s + t_a < n$, among others~\cite{C:BluZhaLos20,AC:BluKatLos21,FC:DelLiu23}, as well as performance \cite{TCC:DelHirLiu21,C:BCLL23}.
The recent work of \cite{EC:DelErb24} has a similar motivation to ours in that the authors show that network-agnostic BA can be built `for free' in the \emph{synchronous} case, namely with $O(\lambda n^2)$ bit and constant round overhead.

\paragraph{Byzantine agreement.} There is a rich history of work on the Byzantine agreement problem in each our considered settings (when considered separately).
In the authenticated setting, the state-of-the-art protocol for BA in terms of communication complexity and latency is \cite{civit2024dare} in the deterministic case, in which they showcase a protocol with resilience $t<\frac{n}{2}$ with optimal $O(f)$ round complexity when $f \leq t$ corruptions actually occur, and $O(\lambda n^2)$ bit complexity.
In the randomized case, \cite{AbrahamDDN019} presents a protocol with $O(1)$ expected latency, and $O(\lambda n^2)$ expected bit complexity, with resilience of $t<\frac{n}{2}$. For the sabotaged setting, the protocol of \cite{MostefaouiMR15} presents a protocol with optimal $t<\frac{n}{3}$ resilience, $O(1)$ expected latency, and $O(n^2)$ expected bit complexity. Alas their protocol assumes the existence of a common coin.
\cite{berman1992bit} and \cite{coan1992modular} were the first to solve BA with $O(n^2)$ bits and linear rounds; later \cite{lenzen2022recursive} built such a protocol that is additionally early stopping.
In a breakthrough result, Chen~\cite{chen2021optimal} solved BA with strong unanimity ($\bot$ can be decided when not all honest parties propose the same value) with $O(nL + n^2 \log{n})$ for messages of length $L$.
\cite{civit2024error} achieve external validity~\cite{C:CKPS01} (decided values satisfy a given predicate) with $O((nL + n^2)\log{n})$ bit complexity; note \name\ can be modified to support external validity by adding appropriate predicate checks.


\section{Preliminaries and Definitions}\label{sec:prelims}
Throughout the paper, we consider a fully connected network of $n$ parties $p_1, \dotsc, p_n$ that communicate over point-to-point authenticated channels.
Some fraction of these parties are controlled by an adversary and may deviate arbitrarily from the protocol.
We call these parties \emph{corrupt} and the other parties \emph{honest}.
When we say that a party \emph{multicasts} a message, we mean that it sends it to all $n$ parties in the network.
We denote the security parameter by $\lambda$.
Throughout the paper, we assume a universe of values $V$.


\paragraph{Public Key Infrastructure.} We assume that the parties have established a public key infrastructure before the protocol execution, which is a bulletin board or plain PKI.
Namely, each party $p_i$ has a secret-public key pair $(\sk_i,\pk_i)$ for the use of cryptography. In this paper, we assume that those keys are used to instantiate a secure digital signature scheme and all messages in our protocols (but not necessarily building blocks) \emph{are implicitly signed}. 


\paragraph{Threshold Signatures.} On top of a PKI, a trusted setup allows the parties to map any vector of $f$ valid signatures of the same message $m$ by different parties (henceforth referred to as an $f$-certificate of $m$) into a single message $\pi$ of length $O(\lambda)$, called a \emph{threshold signature}, denoted $\mathcal{C}(m)$, with the property that the signature certification algorithm passes on $\pi$ iff $\pi$ is the image of a valid $t$-certificate on $m$.

\paragraph{Communication Model.}
We assume a synchronous network, where all parties begin the protocol at the same time, the clocks of the parties progress at the same rate, and all messages are delivered within some known finite time $\Delta > 0$ (called the network delay) after being sent.
In particular, messages of honest parties cannot be dropped from the network and are always delivered.
Thus, we can consider protocols that execute in rounds of length $\Delta$ where parties start executing round $r$ at time $(r-1)\Delta$. We further assume that $\Delta$ is public information and is known to all parties and the adversary, and any action carried out by any party can depend on $\Delta$.
With that in mind, and to avoid notation encumbrance, we omit $\Delta$ from the list of inputs to algorithms and protocols in our definitions.

\paragraph{Adversarial Model.} The adversary model we consider in the paper is an amalgamation of two common adversaries in the literature. Formally, given two parameters $t_i\leq t_s < n/2$ such that $2t_i+t_s<n$, the adversary $\A$ can be described as a tuple $\A=(\A_0,\A_1,\A_2)$ such that 
\begin{itemize}
    \item $\A_0(\Pi,r,Tr_r)=\F_r$ where $\F_r$ denotes the set of corrupt parties at round $r$. I.e. $\A_0$ is an algorithm that chooses for every round the set of corrupt parties, based on the description of the protocol $\Pi$, the round $r$, and the transcript $Tr_r$ of the protocol up to round $t$. We distinguish between two types of adversaries in this context. A \emph{static} adversary satisfies that $\A_0(\Pi,r,Tr_r)=\A_0(\Pi,0,Tr_0)$ for all rounds $r$. An \emph{adaptive} adversary satisfies $\A_0(\Pi,r,Tr_r)\subseteq \A_0(\Pi,r+1,Tr_{r+1})$ for all rounds $t$. Unless otherwise stated, we assume an adaptive adversary 
    For a given adversary $\A$, we say that a party $p$ is \emph{forever honest} if $p\not\in \F_r$ for all rounds $r$.
    \item $\A_1(\Pi,r,Tr_r,\F_r)$ describes the algorithm run by corrupt parties throughout the execution of the protocol: it may depend on the description of the protocol $\Pi$, the round $r$, the transcript $Tr_r$ of the protocol up to round $r$, and the internal state of all corrupt parties at round $r$. In this context, we distinguish between two settings, characterized by the capabilities of the adversary.
        \begin{itemize}
            \item \textbf{Sabotaged.} $\A_1$ (and $\A_0$) are computationally unbounded, and in particular can break the security of any cryptographic primitive used in the protocol via the PKI. Furthermore, the adversary has complete access to all the information, secret and public of any setup protocol carried out by the parties prior to receiving their inputs. Equivalently, in the ideal setup world, the adversary receives from the trusted dealer of the setup all communication sent to any party. 
            \item \textbf{Authenticated.} $\A_1$ (and $\A_0$) are computationally bounded, and there is a trusted PKI setup. 
            In this case, we assume in our security proofs that the cryptographic primitives used in the protocol provide perfect security, which, by a standard hybrid argument, does not affect the generality of our result and serves to simplify the exposition.
        \end{itemize}
    \item $\A_2(\Pi)\to \set{0,1}$. The adversary, at round $r=0$ can view the protocol description $\Pi$ and choose a bit $b$ that indicates the \emph{setting} of the current execution of the protocol. This choice is not revealed to honest parties. The following holds.
        \begin{itemize}
            \item If $b=1$, then $\A$ chose the sabotaged setting. Furthermore, $|\F_r|\leq t_i$ for all rounds $r$.
            \item If $b=0$, then $\A$ chose the authenticated setting. Furthermore $|\F_r|\leq t_s$ for all rounds $r$.
        \end{itemize}
    
\end{itemize}

We say that the adversary $\A$ is $t$-bounded if $|\F_r|\leq t$ holds for all rounds $r$ 

\noindent Definitions and properties that we introduce hereafter are only required to hold with probability $1 - \negl[\lambda]$.

\subsection{Distributed Primitives}

When relevant, our primitives take input from a value set $V$ with $|V| \geq 2$; we assume that default value $\bot \not\in V$.
Note that $\bot$ is considered as a valid output in each protocol.


\begin{definition}[Byzantine Agreement]\label{def:BA}
Let $\Pi$ be a protocol executed by parties $p_1, \dots, p_n$, where each party $p_i$ begins by calling $\mathsf{propose}$ with input $v_i \in V$. The $\BA$ problem pertains to the following properties.
\begin{itemize}
\item \emph{\textbf{Validity:}} $\Pi$ is sabotaged (authenticated) $t$-$\mathsf{valid}$ if the following holds in the sabotaged (authenticated) setting when at most $t$ parties are corrupted: If every honest party's input is equal to the same value $v$, then every honest party outputs $v$.
\item \emph{\textbf{Consistency:}} $\Pi$ is sabotaged (authenticated) $t$-$\mathsf{consistent}$ if the following holds in the sabotaged (authenticated) setting when at most $t$ parties are corrupted: Every honest party that outputs a value outputs the same value $v$.
\item \emph{\textbf{Termination:}} $\Pi$ is sabotaged (authenticated) $t$-$\mathsf{terminating}$ if the following holds in the sabotaged (authenticated) setting when at most $t$ parties are corrupted: Every honest party produces an output and terminates.
\end{itemize}
If $\Pi$ is sabotaged (authenticated) $t$-valid, $t$-consistent, and $t$-terminating, we say it is sabotaged (authenticated) $t$-$\mathsf{secure}$.
\end{definition}

\begin{definition}[Graded Consensus]\label{def:GC}
    In the graded consensus ($\GC$) problem, each honest party invokes $\mathsf{propose}$ with input $v_i\in V$ and outputs a tuple $(y_i,g_i)\in V \times \set{0,1}$. Let $\Pi$ be a protocol executed by parties $p_1,...,p_n$. The relevant properties attributable to $\Pi$ are as follows.
    \begin{itemize}
        \item \textbf{Validity:} We say that $\Pi$ is sabotaged (\emph{authenticated}) $t$-valid if the following holds in the sabotaged (authenticated) setting when at most $t$ parties are corrupted: If every honest party's input is equal to the same value $v$, then every honest party outputs $(v,1)$. 
        \item \textbf{Consistency:} We say that $\Pi$ is sabotaged (\emph{authenticated}) $t$-consistent if the following holds in the sabotaged (authenticated) setting when at most $t$ parties are corrupted: If an honest party outputs $(v,1)$ for some value $v$, then all honest parties output either $(v,1)$ or $(v,0)$.
        \item \textbf{Termination:} We say that $\Pi$ is sabotaged (\emph{authenticated}) $t$-terminating if the following holds in the sabotaged (authenticated) setting when at most $t$ parties are corrupted: There exists a round $r$ such that all honest parties produce an output and terminate by round $r$.
    \end{itemize}
    If a protocol $\Pi$ for GC is sabotaged (\emph{authenticated}) $t$-valid, $t$-consistent, and $t$-terminating, we say that $\Pi$ is sabotaged (\emph{authenticated}) $t$-secure.
\end{definition}

\begin{definition}[Synchronizer]\label{def:synchro}
    In the Synchronizer problem, we expose the following interface, that any party can engage with in a round of their choice.
    \begin{itemize}
    \item $\mathsf{start\_synchronization}(v \in V)$: a party starts synchronization with a value $v \in V$.

    \item output $\mathsf{synchronization\_completed}(v' \in V)$: a party completes synchronization with a value $v' \in V$.
     \end{itemize}

     We make the assumption that each honest party starts synchronization at most once. Importantly, we do \emph{not} assume that all honest parties start synchronization, i.e., it could be the case that no honest party starts synchronization.

     We consider the following properties w.r.t. to a protocol $\Pi$ in the context of the Synchronization primitive.
     \begin{itemize}
         \item \textbf{Justification:} We say that $\Pi$ has sabotaged (authenticated) $t$-justification if the following holds in the sabotaged (authenticated) setting when at most $t$ parties are corrupted: If the an honest party $p$ completes synchronization with a value $v'$ at round $r$, then there exists an honest party $q$ that started synchronization with value $v'$ at a round $r'<r$.
         \item \textbf{Totality:} We say that $\Pi$ has sabotaged (authenticated) $t$-totality if the following holds in the sabotaged (authenticated) setting when at most $t$ parties are corrupted: Let $\rho$ be the first round in which an honest party completes synchronization for some value $v$. Then, every honest party $p_i$ completes synchronization at some round $\rho_i\leq \rho+1$.
         \item \textbf{Liveness:} We say that $\Pi$ has sabotaged (authenticated) $t$-liveness if the following holds in the sabotaged (authenticated) setting when at most $t$ parties are corrupted: Suppose there exists a value $v\in V$ and a round $\rho$ such that all honest parties start synchronization with value $v$ by round $\rho$. Then, every honest party $p_i$ completes synchronization with value $v$ at some round $\rho_i\le \rho+1$.
     \end{itemize}
\end{definition}


\paragraph{Generic Compiler.} As mentioned in the introduction, we make use of black-box access to given protocols solving BA. One in the PKI setting with  authenticated $t_s$-security, and one in the information theoretic setting with $t_i$-security. We further make the assumption that along with $\Pi$, we are also given a parameter $T_\Pi$ indication the amount of rounds one must run the protocol to ensure that all honest nodes have produced an output except for with negligible probability We here formalize the notion of black-box access to a protocol in the context of our work.

\begin{definition}\label{def:blackbox_Access}
    For a given tuple $(\Pi,T_\Pi)$, where $\Pi$ is a sabotaged (authenticated) $t$-secure $\BA$ protocol and $T_\Pi$ is a parameter, black-box access implies the following guarantees for any adversary $\A$.
    \begin{itemize}
        \item If $\A$ chose the sabotaged (authenticated) setting and furthermore there exists a round $r'$ s.t. all honest parties initiate $\Pi.\mathsf{propose(v\in V)}$ at round $r'$, then except for with $\negl$ probability, by round $r'+T_\Pi$, all honest parties have produced an output from $\Pi$, furthermore, these outputs satisfy the $\BA$ conditions so long as $|\F_r|\leq t$ for all rounds $r$.
    \end{itemize}
\end{definition}
For a given protocol $\Pi$ and an adversary $\A$, we denote its expected communication (in bits) complexity under $\A$ by $\CC_\A(\Pi)$. Note that this is well defined since in all protocols discussed in this work, there is a predetermined upper bound on the number of rounds for which each party is active before halting.
We assume that values in $V$ are of size $O(\lambda)$ when $\CC$ is calculated; this is without loss of generality (see \cref{sec:conclusion}).

\section{\name}\label{sec:jaggernaut}

This section presents \name, our main protocol.
We start by introducing \name's building blocks (\Cref{subsection:building_blocks_overview}).
Then, we show how these building blocks are composed into \name (\Cref{subsection:name_implementation}).
Finally, we prove \name's security and complexity, captured in the following theorem.
\begin{theorem}\label{thrm:main_theorem}
Let $t_s,t_i$ such $2t_i+t_s<n,t_i\leq t_s<\frac{n}{2}$. Assuming black-box access to $(\Piauth,T_\auth)$ and $(\Pisab,T_\Sab)$, where $\Piauth$ is an authenticated $t_s$-secure protocol for $\BA$, and $\Pisab$ is a sabotaged $t_i$-secure protocol. Then, there exists a protocol \name which is authenticated $t_s$-secure \emph{and} sabotaged $t_i$-secure. Furthermore, for any adversary $\A$ the following holds:
\begin{enumerate}
    \item If $\Piauth$ and $\Pisab$ are secure against an adaptive adversary, then so is \name.
    \item If $\A$ chose the sabotaged setting, and the adversary is $t_i$-bounded, the expected communication (in bits) complexity of \name is $O(\CC_\A(\Piauth)+\CC_\A(\Pisab)+\lambda n^2)$.
    \item If $\A$ chose the authenticated setting, and the adversary is $t_s$-bounded, then the expected communication (in bits) complexity of \name is $O(\CC_\A(\Piauth)+\lambda n^2)$. 
    \item If $\A$ chose the authenticated setting, and the adversary is $t_s$-bounded, then if $r$ is the the first round such that all parties honest at round $r$ have produced an output from $\Piauth$ and terminated $\Piauth$, then all forever honest parties produce and output from \name and terminate \name after $r+O(1)$ rounds. 
\end{enumerate}
    
\end{theorem}

\subsection{Building Blocks: Overview} \label{subsection:building_blocks_overview}
In this subsection, we summarise the building blocks we use to build \name on top of Byzantine agreement.
\subsubsection{Authenticated Graded Consensus with Sabotaged Validity.}

The authenticated graded consensus with fallback validity primitive exposes the following interface:
\begin{itemize}
    \item Input $\mathsf{propose}(v \in V)$: A party proposes a value $v \in V$.

    \item Output: $(v' \in V, g' \in \{ 0, 1\})$: A party outputs a value $v'$ with a binary grade $g'$. Usually indicated by a left arrow $\leftarrow$.
\end{itemize}
We assume that all honest parties propose exactly once and they all do that simultaneously (i.e., in the same round). In this setting, we design a protocol $\GCone$, that satisfies the following properties w.r.t. the graded consensus primitive (See \cref{def:GC}) for any $t_s+2t_i<n,t_i\leq t_s<\frac{n}{2}$: Authenticated $t_s$-secure, sabotaged $t_i$-valid, and sabotaged $t_i$-terminating.



\paragraph{Complexity.}
\name utilizes an implementation of the primitive that exchanges $O(\lambda n^2)$ bits and terminates in $T_1=O(1)$ rounds.
We relegate our implementation of the primitive to \cref{ssec:GC_PKI_star}.

\subsubsection{Sabotaged Graded Consensus with Authenticated Validity.}

The sabotaged graded consensus with authenticated validity primitive exposes the following interface:
\begin{itemize}
    \item Input $\mathsf{propose}(v \in V)$: A party proposes a value $v \in V$.

    \item Output $(v' \in V, g' \in \{ 0, 1 \})$: A party outputs a value $v'$ with a binary grade $g'$. Usually indicated with a left arrow $\leftarrow$.
    
\end{itemize}
All honest parties propose exactly once and they do so within one round of each other.
Therefore, we do \emph{not} assume that all honest parties propose in the same round. For this setting, we design a protocol $\GCtwo$ with the following properties w.r.t. to the $\GC$ primitive (See \cref{def:GC}) for any $t_s+2t_i<n,t_i\leq t_s<\frac{n}{2}$: Sabotaged $t_i$-secure and authenticated $t_s$-valid and $t_s$-terminating.



\paragraph{Complexity.}
In \name, we employ an implementation of the primitive that exchanges $O(\lambda n^2)$ bits and terminates in $T_2=O(1)$ rounds.
The implementation can be found in \Cref{ssec:GC_Sab_star}.

\subsubsection{Synchronizer.}
The primitive exposes the following interface:
\begin{itemize}
    \item input $\mathsf{start\_synchronization}(v \in V)$: A party starts synchronization with a value $v \in V$.

    \item output $\mathsf{synchronization\_completed}(v' \in V)$: A party completes synchronization with a value $v' \in V$.
\end{itemize}
We design a protocol $\Sync$, that has the following properties in the context of the Synchronizer primitive (See \cref{def:synchro}) for any $t_s+2t_i<n,t_i\leq t_s<\frac{n}{2}$.
\begin{itemize}
    \item Sabotaged $t_i$-totality.
    \item Authenticated $t_s$-justification, $t_s$-totality, $t_s$-liveness.
\end{itemize}




\paragraph{Complexity.}
We implement the synchronizer primitive with $O(\lambda n^2)$ exchanged bits. See \cref{ssec:synchronizer} for more details.

\subsection{\name's Implementation \& Proof} \label{subsection:name_implementation}

\name's implementation is provided in \Cref{fig:name}, whereas its visual dedication can be found in \Cref{fig:name_visual}. For clarity, we proceed with a written exposition on the stages of the protocol. On a high level, as seen in \cref{fig:name_visual}, the protocol divided into five steps.

First, all parties propose their input into $\GCone$ -- recall that this is authenticated $t_s$-secure and sabotaged $t_i$-valid.
The output of $\GCone$ is then fed as input into $\Piauth$.
Note that here is where the synchronizer $\Sync$ comes into play, since in the the sabotaged case, or an early stopping/randomized protocol $\Piauth$ in the authenticated setting, the adversary can cause significant gaps between the rounds in which honest parties produce an output and move on from $\Piauth$.
$\Sync$ maintains that honest parties exit $\Piauth$ at most 1 round apart from one another. The input to $\GCtwo$ is then decided depending on whether $g_1=1$, as depicted in the figure.

To maintain early stopping in the authenticated case, each party that has $g_2=1$ participates in an early stopping phase in order to detect pre-existing agreement in $\GCtwo$. During that phase, parties with $g_2=1$ multicast $\decide(v_4)$ messages for their output from $\GCtwo$. If a quorum of $n-t_s$ $\decide(v)$ messages is received for some value, then a party decides $v$ and halts. Otherwise, after sufficient time has passed, all honest parties that haven't halted start executing the sabotaged BA protocol $\Pisab$ in the same round. At its conclusion, parties produce an output based on conditions C2 and C3, as seen in \cref{fig:name}.

\begin{figure}[ht!] 
\centering
\fbox{
\parbox{0.95\textwidth}{
\centering \name \\ (Pseudocode for a party $p_i$) \\
\textbf{Constants:} $T_{\max} = T_1 + T_\auth + T_2 + 1$. \\
\textbf{Initialization:} $v_i = p_i$'s proposal, $\output_1 = \output_2 = \false$.
\begin{enumerate}
\item Let $(v_1, g_1) \gets \GCone.\mathsf{propose}(v_i)$. \textcolor{blue}{// This step takes exactly $T_1$ rounds.}
\item Let $v_2 \gets \Piauth.\mathsf{propose}(v_1)$.
If $\Piauth$ does not terminate after $T_\auth$ rounds, let $v_2 \gets \bot$. \textcolor{blue}{// This step takes at most $T_\auth$ rounds. However, if $\mathcal{A}$ is computationally bounded, it might take fewer than $T_\auth$ rounds.}
\item Invoke $\Sync.\mathsf{start\_synchronization}(v_2)$.
\item Upon (1) $\Sync.\mathsf{synchronization\_completed}(v')$ is triggered, or (2) $T_\auth$ rounds elapsed, party $p_i$ starts step 4.
If $g_1 = 1$, let $v_3 \gets v_1$.
Else if $\Sync.\mathsf{synchronization\_completed}(v')$ is triggered, let $v_3 \gets v'$.
Else, let $v_3 \gets v_1$.
\item Let $(v_4, g_4) \gets \GCtwo.\mathsf{propose}(v_3)$. \textcolor{blue}{// This step takes exactly $T_2$ rounds.}
\item Set $\output_1 = \true$.
\item If $g_4 = 1$, multicast $(\decide, v_4)$.
\item At round $T_{\max}$, let $\mathsf{output}_1 = \false$ and let $v_\Sab \gets \Pisab.\mathsf{propose}(v_4)$.
If $\Pisab$ does not terminate after $T_\Sab$ rounds, let $v_\Sab \gets \bot$.
\item Set $\output_2 = \true$.
\end{enumerate}
\textbf{Decision:}
\begin{itemize}
\item[C1.] If $\output_1 = \true$ and received $(\decide, v_4)$ from $n - t_s$ distinct parties: decide $v_4$ and halt.
\item[C2.] If $\output_2 = \true$ and received for some $v$ $(\decide, v) $ from $n - t_s - t_i$ distinct parties: decide $v$ and halt.
\item[C3.] If $\output_2 = \true$: decide $v_\Sab$ and halt.
\end{itemize}
}
}
\caption{BA with crypto-agnostic security for $2 t_i + t_s < n$ given 
1) an authenticated $t_s$-secure BA protocol $\Piauth$; 
2) a sabotaged $t_i$-secure BA protocol $\Pisab$;
3) an authenticated $t_s$-secure and sabotaged $t_i$-valid graded consensus protocol $\GCone$; and 
4) a sabotaged $t_i$-secure and authenticated $t_s$-valid graded consensus protocol $\GCtwo$. 
}
\label{fig:name}
\end{figure}

\begin{figure}
    \centering{
    \includegraphics[width=.98\textwidth, left]{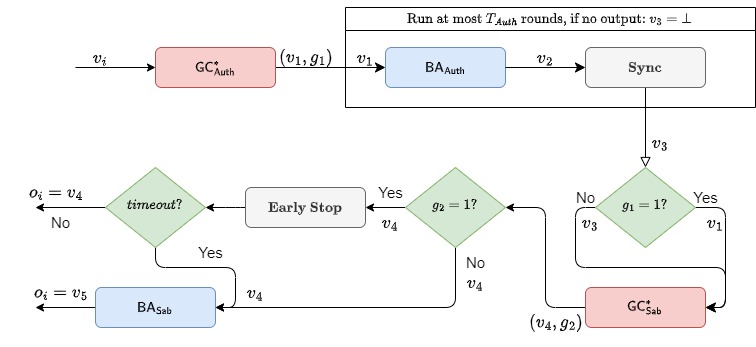}}
    \caption{Overview of the structure of the main protocol for each party. Beginning with input insertion at the top left with $v_i$, and ending with output production at the bottom left. The values next to arrows indicate the outputs and inputs produced and delivered into procedures.}
    \label{fig:name_visual}
\end{figure}

\begin{proof}[of \cref{thrm:main_theorem}]

\paragraph{\name's security in the sabotaged setting.}
We start by proving that \name is $t_i$-$\mathsf{valid}$ in the sabotaged setting.

\begin{lemma} [$t_i$-validity] \label{lemma:main_no_sig_validity}
\name (\Cref{fig:name}) is $t_i$-$\mathsf{valid}$ in the sabotaged setting.
\end{lemma}
\begin{proof}
Suppose all honest parties propose the same value $v$.
Due to the validity property of $\GCone$ in the sabotaged setting, all honest parties decide $(v, 1)$ from $\GCone$.
Hence, all honest parties propose $v$ to $\GCtwo$.
The fact that $\GCtwo$ satisfies validity in the sabotaged setting proves that all honest parties decide $(v, 1)$ from $\GCtwo$.
Thus, all honest parties multicast a $(\mathsf{decide}, v)$ message.
Hence, all honest parties receive $n - t_i$  such messages.
(Note that no honest party receives $n - t_s$ $\mathsf{decide}$ message for any value $v' \neq v$ since that would imply that an honest party sends a $\mathsf{decide}(v')$ message, which cannot occur.)
As $n - t_i \geq n - t_s$ (since $t_i \leq t_s$), all honest parties decide $v$ according to decision condition C1.
\qed
\end{proof}
Next, we prove that \name is $t_i$-$\mathsf{consistent}$ in the sabotaged setting.

\begin{lemma} [$t_i$-consistency]
\name (\Cref{fig:name}) is $t_i$-$\mathsf{consistent}$ in the sabotaged setting.
\end{lemma}
\begin{proof}
Suppose that at least one honest party outputs due to C1, i.e., receives $(\decide, v_4)$ from $n - t_s$ distinct parties.
Since $n - t_s - t_i > t_i$, at least one honest party must have output $(v_4, 1)$ from $\GCtwo$, and thus by $t_i$-consistency of $\GCtwo$, all honest parties output the same value $v_4$, and in particular no honest party multicasts $(\decide, v_4')$ for $v_4' \neq v_4$.
In addition, at least $n - t_s - t_i$ honest parties must have multicast $(\decide, v_4)$.
Thus all honest parties will receive $(\decide, v_4)$ from at least $n - t_s - t_i$ distinct parties, and not for any other value $v_4' \neq v_4$ since $n - t_s - t_i > t_i$.
Therefore all honest parties will output $v_4$ either due to C1 or C2.

Suppose that no honest party outputs due to C1 and at least one outputs due to C2, i.e., due to receiving $(\decide, v)$ for some $v$ from $n - t_s - t_i$ distinct parties.
Then at least one honest party must have multicast $(\decide, v)$, since $n - t_s - t_i > t_i$.
By the $t_i$-consistency of $\GCtwo$, all honest parties will propose the same $v$ to $\Pisab$.
By the $t_i$-validity of $\Pisab$, all honest parties will output the same $v_\Sab = v$.
Thus all honest parties will output either due to C2 or C3 with the same value $v$.

Finally, supposing no honest party outputs due to C1 or C2, all honest parties output the same value $v$ from $\Pisab$ due to the $t_i$-consistency of $\Pisab$ and then decide it due to C3.
\qed
\end{proof}
Finally, we prove that \name is $t_i$-$\mathsf{terminating}$ in the sabotaged setting.

\begin{lemma} [$t_i$-termination]
\name (\Cref{fig:name}) is $t_i$-$\mathsf{terminating}$ in the sabotaged setting.
\end{lemma}
\begin{proof}
\name trivially terminates as all honest parties decide in round $T_{\max} + T_\Sab$ (at the latest).
\qed
\end{proof}

\paragraph{\name's security in the authenticated setting.}
First, we prove that $\name$ is $t_s$-$\mathsf{valid}$ in the authenticated setting.

\begin{lemma} [$t_s$-validity] \label{lemma:main_sig_validity}
\name (\Cref{fig:name}) is $t_s$-$\mathsf{valid}$ in the authenticated setting. Furthermore, all honest parties output according to C1.
\end{lemma}
\begin{proof}
Let all honest parties propose the same value $v$.
By the $t_s$-security of $\GCone$, all honest parties output $(v, 1)$.
Then as above, since all honest parties output $g_1 = 1$, all honest parties set $v_3 = v$.
By the $t_s$-validity with termination of $\GCtwo$, all parties output $(v, 1)$.
Validity and all parties outputting according to C1 then follows as in the proof of~\cref{lemma:main_no_sig_validity}
\qed
\end{proof}
Next, we prove that \name is $t_s$-$\mathsf{consistent}$ in the authenticated setting.

\begin{lemma} [$t_s$-consistency]\label{lemma:main_sig_consistency}
\name (\Cref{fig:name}) is $t_s$-$\mathsf{consistent}$ in the authenticated setting. Furthermore, all honest parties output according to C1.
\end{lemma}
\begin{proof}
We first argue that all honest parties input the same value $v_3$ to $\GCtwo$. 
\begin{itemize}
    \item Suppose first that at least one honest party outputs $g_1 = 1$ from $\GCone$. Then by the $t_s$-consistency of $\GCone$, all honest parties will output the same $v_1$. By the $t_s$-validity of $\Piauth$, all honest parties output $v_2 = v_1$ from $\Piauth$, and thus no honest party can advance from $\Sync$ with any value $v'\neq v_1$. Thus $v_3=v_1$ is the same for all honest parties.
    \item Otherwise, no party outputs $g_1 = 1$ from $\GCone$. By $t_s$-consistency of $\Piauth$, all honest parties output the same $v_2$ (and thus no honest party can advance from $\Sync$ with any value $v'\neq v_2$), and by construction all set $v_3 = v_2$ (i.e., do not override $v_3$ with $v_1$).
\end{itemize}
Then, by the $t_s$-validity of $\GCtwo$, all honest parties output $(v_4, 1)$ from $\GCtwo$, multicast $(\decide, v_4)$, and by the same reasoning as for the proof of~\cref{lemma:main_sig_validity}, all honest parties will output due to C1 (and in particular never execute $\Pisab$).
\qed
\end{proof}
Lastly, we observe that \name is $t_s$-$\mathsf{terminating}$ in the authenticated setting as all parties terminate without running $\Pisab$.


\paragraph{Complexity.} We now attend to the complexity analysis of \name. \cref{lemma:main_sig_validity} and \cref{lemma:main_sig_consistency} implies that in the authenticated setting, all honest parties produce an output and halt after $\GCtwo$, without running $\Pisab$. Since $T_2=O(1)$, and by the Totality and Liveness properties of $\Sync$, we get that if $r$ is first the round by which all honest parties produced an output from $\Piauth$, then all forever honest parties produce an output and halt after at most $r+O(1)$ rounds, as required. Furthermore, since no honest party executes $\Pisab$, the expected communication (in bits) complexity of \name is $O(\CC_\A(\Piauth)+\lambda n^2)$, as required. In the sabotaged setting, notice that all communication aside from $\Piauth$ and $\Pisab$ is bounded by $O(\lambda n^2)$ bits.
Thus in the sabotaged setting, the expected communication (in bits) complexity of \name is $O(\CC_\A(\Piauth)+\CC_A(\Pisab)+\lambda n^2)$, as required. 

\noindent Lastly, notice that aside from (potentially) $\Piauth$ and $\Pisab$, the protocol \name is deterministic, and thus if $\Piauth,\Pisab$ are secure against an adaptive adversary, then so is \name. This concludes the proof of \cref{thrm:main_theorem}.
\qed
\end{proof}

\subsection{Corollaries}\label{ssec:corols}

Now that we have proven \cref{thrm:main_theorem}, we can instantiate $\Piauth$ and $\Pisab$ with concrete protocols in order to obtain concrete crypto-agnostic protocols.
Before we state those corollaries, one point needs to be addressed.

\paragraph{Bit complexity in the sabotaged setting.} In order to maximize the generality of our results, the only assumption we made about $\Piauth,\Pisab$, as explained in \cref{def:blackbox_Access}, is that we are provided with round complexity bounds for these protocols. In particular, we are given no such guarantee for bit complexity. As such, in the \emph{sabotaged} setting, when the parties run $\Piauth$, which is an \emph{authenticated} $t_s$-secure $\BA$ protocol, we have no a-priori upper bound for the amount of bits sent by honest parties during the execution, hence why we define the bit complexity of a protocol w.r.t. a particular adversary.
When considering concrete protocols, however, each party can infer a bound from the description of the protocol on the number of messages it has to send, and if is exceeded, an honest party can simply halt $\Piauth$ and move onto $\Sync$ with input $\bot$.
The corollaries we state assume that such a modification was made to the final protocol.

\paragraph{$\name_\mathsf{det}$.} For a deterministic protocol, we instantiate $\Piauth$ with the protocol of \cite{momose2021optimal} modified using the techniques of~\cite{lenzen2022recursive} to achieve $O(f)$ round complexity, and $\Pisab$ with the protocol of \cite{berman1992bit}, to obtain the following.

\begin{corollary}\label{corol:protocol_deterministic}
    Let $t_s,t_i$ s.t. $t_s+2t_i<n,t_i\leq t_s<\frac{n}{2}$. There exists a deterministic authenticated $t_s$-secure, sabotaged $t_i$-secure protocol $\name_{\mathsf{det}}$ solving the BA problem with the following properties.
    \begin{itemize}
        \item In the authenticated setting, $\name_{\mathsf{det}}$ has $O(\lambda n^2)$ bit complexity and $O(f)$ round complexity, where $f \leq t_s$ is number of \emph{actual} corruptions.
        \item In the sabotaged setting, $\name_{\mathsf{det}}$ has $O(\lambda n^2)$ bit complexity and $O(n)$ round complexity.
    \end{itemize}
\end{corollary}

\paragraph{$\name_\mathsf{ran}$.} For a randomized protocol, we instantiate $\Piauth$ with the protocol of \cite{AbrahamDDN019}, setting $T_\auth = O(\lambda)$, and $\Pisab$ with the protocol of \cite{lenzen2022recursive} (run bit-by-bit in parallel), to obtain the following.

\begin{corollary}\label{corol:protocol_randomized}
    Let $t_s,t_i$ s.t. $t_s+2t_i<n,t_i\leq t_s<\frac{n}{2}$. There exists a randomized authenticated $t_s$-secure, sabotaged $t_i$-secure protocol $\name_{\mathsf{ran}}$ solving the BA problem with the following properties.
    \begin{itemize}
        \item In the authenticated setting, $\name_{\mathsf{ran}}$ has $O(\lambda n^2)$ expected bit complexity and $O(1)$ expected round complexity.
        \item In the sabotaged setting, $\name_{\mathsf{ran}}$ has $O(\lambda^2 n^2)$ bit complexity and ${O(\lambda + f)}$ round complexity, where $f \leq t_i$ is the number of actual corruptions.
    \end{itemize}
\end{corollary}
\section{Building a Synchronizer}\label{ssec:synchronizer}

In this section we construct $\Sync$, that fills the role of the synchronizer as discussed in \cref{sec:jaggernaut}. Formally, we prove the following.

\begin{theorem}\label{lemma:Sync}
    There exists a deterministic protocol $\Sync$ that satisfies the following in the context of the Synchronizer primitive (See \cref{def:synchro}), for any $t_s+2t_i<n, t_i\leq t_s<\frac{n}{2}$.
    \begin{itemize}
        \item Sabotaged $t_i$-totality.
        \item Authenticated $t_s$-totality, $t_s$-justification, and $t_s$-liveness.
    \end{itemize}
    Furthermore, $\Sync$ has communication (in bits) complexity of $O(\lambda n^2)$.
\end{theorem}
The implementation of $\Sync$ can be found in \cref{fig:synrhonizer}.
On calling $\mathsf{start\_synchronization}$, parties multicast $\mathsf{finish}(v_i)$.
On receiving $n - t_s$ $\mathsf{finish}(v)$ for some $v$, parties form a certificate.
Then, at this point or on receipt of such a certificate, the caller multicasts it and calls $\mathsf{synchronization\_completed}(v)$.




\begin{figure}[ht!]
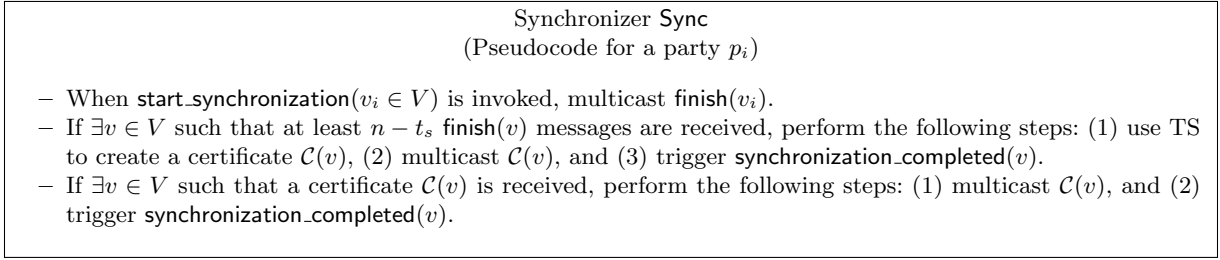
 
\centering
\fbox{
\parbox{0.95\textwidth}{
\centering Synchronizer $\Sync$\\ (Pseudocode for a party $p_i$) 
\begin{itemize}
\item When $\mathsf{start\_synchronization}(v_i \in V)$ is invoked, multicast $\mathsf{finish}(v_i)$.

\item If $\exists v \in V$ such that at least $n - t_s$ $\mathsf{finish}(v)$ messages are received, perform the following steps: (1) use TS to create a certificate $\mathcal{C}(v)$, (2) multicast $\mathcal{C}(v)$, and (3) trigger $\mathsf{synchronization\_completed}(v)$.

\item If $\exists v \in V$ such that a certificate $\mathcal{C}(v)$ is received, perform the following steps: (1) multicast $\mathcal{C}(v)$, and (2) trigger $\mathsf{synchronization\_completed}(v)$.
\end{itemize}
}
}
\caption{Synchronizer $\Sync$ with authenticated $t_s$-justification, totality and liveness, and sabotaged $t_i$-totality.}
\label{fig:synrhonizer}
\end{figure}
\begin{proof}[of \cref{lemma:Sync}]

\paragraph{$\Sync$'s correctness in the authenticated setting.}
First, prove the justification property.

\begin{lemma} [$t_s$-justification]
$\Sync$ (\Cref{fig:synrhonizer}) satisfies $t_s$-justification in the authenticated setting.
\end{lemma}
\begin{proof}
If an honest party completes synchronization with a value $v'$, there is at least one honest party that multicasts a $\mathsf{finish}(v')$ message (as $n - t_s > t_s$ given $t_s < n / 2$).
Hence, the lemma holds.
\qed
\end{proof}
Next, we prove the totality property.

\begin{lemma} [$t_s$-totality]
$\Sync$ (\Cref{fig:synrhonizer}) satisfies $t_s$-totality in the authenticated setting.    
\end{lemma}
\begin{proof}
The lemma trivially holds as each honest party disseminates a certificate once it completes synchronization.
\qed
\end{proof}
Finally, we prove the liveness property.

\begin{lemma} [$t_s$-liveness]
$\Sync$ (\Cref{fig:synrhonizer}) satisfies $t_s$-liveness in the authenticated setting.
\end{lemma}
\begin{proof}
By the beginning of the round $\rho + 1$, each honest party receives $n - t_s$ $\mathsf{finish}(v)$ messages.
Therefore, each honest party completes synchronization with value $v$ by round $\rho + 1$.
\qed
\end{proof}

\paragraph{$\Sync$'s correctness in the sabotaged setting.}

We prove the totality property.

\begin{lemma} [$t_i$-totality]
$\Sync$ (\Cref{fig:synrhonizer}) satisfies $t_i$-totality in the sabotaged setting.
\end{lemma}
\begin{proof}
The lemma trivially holds as each honest party disseminates a certificate once it completes synchronization (as in the authenticated case).
\qed
\end{proof}

\paragraph{Complexity.}
Finally, we argue that honest parties exchange $O(\lambda n^2)$ bits in $\Sync$.
Observe that each honest party multicasts a single $\mathsf{finish}$ message, thus sending $O(n)$ bits (assuming that the values are of constant size).
Moreover, each honest party multicasts one certificate of size $O(\lambda)$ bits, thus sending $O(\lambda n)$ bits.
Hence, honest parties send $n \cdot O(n + \lambda n) = O(\lambda n^2)$ bits.
\qed
\end{proof}

\section{Building Graded Consensus}\label{sec:Building Blocks}
In this section, we construct our two $\GC$ protocols that we use in \name\ -- First, our authenticated protocol with sabotaged validity and termination, and then, our sabotaged protocol with authenticated validity and termination.
\subsection{Authenticated Graded Consensus with Sabotaged Validity}\label{ssec:GC_PKI_star}
We first construct graded consensus for the authenticated setting, augmented with validity for the sabotaged setting.
We denote the protocol by $\GCone$. Formally, we prove the following.

\begin{theorem}\label{lemma:GC_PKI_star}
    Let $t_s,t_i$ such that $t_s+2t_i<n,t_i\leq t_s<\frac{n}{2}$. There exists a $T_1=O(1)$ round deterministic protocol ${\GCone}$ that satisfies the following properties in the context of the $\GC$ primitive (see \cref{def:GC}) in a $\Delta$-synchronous network when all parties commence the protocol in the same round.
    \begin{itemize}
        \item Authenticated $t_s$-secure.
        \item Sabotaged $t_i$-valid.
        \item Sabotaged $t_i$-terminating.
    \end{itemize}
Furthermore, given $\GCPKI$ has at most $O(\lambda n^2)$ bit complexity and constant round complexity, so too does $\GCone$.
\end{theorem}
Our protocol $\GCone$, described in \cref{fig:gc_1_jovan}, assumes an authenticated graded consensus protocol (see \cref{def:GC}), like the constant-round, authenticated $t_s$-secure, and sabotaged $t_i$-terminating GC  of Momose-Ren~\cite{Momose2021}.
In $\GCone$, parties run three rounds of filtering to additionally ensure sabotaged $t_i$-validity before executing $\GCPKI$.
First, each party multicasts its input (recall all values are signed by assumption).
Then, if $n - t_i$ signatures from different parties on some $v$ are received, a certificate on $v$ is formed, the output of the protocol is locked to $(v, 1)$, $p_i$'s input to $\GCPKI$is overwritten with $v$, and the certificate is multicast.
In the third round, parties overwrite their input if they receive non-conflicting certificates for a single value $v$.
Finally, parties run $\GCPKI$, and output the result of that if they did not lock their output.

\begin{figure}
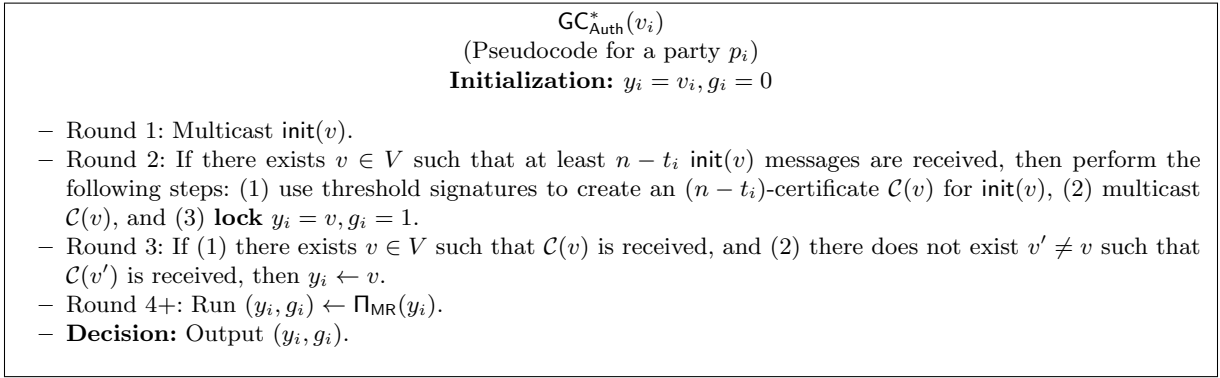

\centering
\fbox{
\parbox{0.95\textwidth}{
\centering ${\GCone}(v_i)$\\
(Pseudocode for a party $p_i$)\\
\textbf{Initialization: $y_i=v_i,g_i=0$}
\begin{itemize}

\item Round 1: Multicast $\mathsf{init}(v)$.

\item Round 2: If there exists $v \in V$ such that at least $n - t_i$ $\mathsf{init}(v)$ messages are received, then perform the following steps: (1) use threshold signatures to create an $(n - t_i)$-certificate $\mathcal{C}(v)$ for $\mathsf{init}(v)$, (2) multicast $\mathcal{C}(v)$, and (3) \textbf{lock} $y_i=v,g_i=1$.

\item Round 3: If (1) there exists $v \in V$ such that $\mathcal{C}(v)$ is received, and (2) there does not exist $v' \neq v$ such that $\mathcal{C}(v')$ is received, then $y_i \gets v$.

\item Round 4+: Run $(y_i,g_i) \gets \GCPKI(y_i)$.

\item \textbf{Decision:} Output $(y_i,g_i)$.

\end{itemize}

}
}
\caption{Authenticated $t_s$-secure and sabotaged $t_i$-valid graded consensus protocol $\GCone$ given an authenticated $t_s$-scure graded consensus $\GCPKI$.}

\label{fig:gc_1_jovan}
\end{figure}

\begin{proof}[of \cref{lemma:GC_PKI_star}]
     Termination clearly holds in both settings by the behaviour of the protocol. We now move on to the other properties.

\paragraph{Validity in the sabotaged setting.}
First, we prove that ${\GCone}$ satisfies validity in the sabotaged setting.

\begin{lemma} \label{lemma:init_correct_proposal}
If any honest party receives $n - t_i$ $\mathsf{init}(v)$ messages, then $v$ is the proposal of an honest party.
\end{lemma}
\begin{proof}
As $n > t_s + 2t_i$, we have that $n - t_i > t_i$.
Therefore, the lemma holds.
\qed
\end{proof}

\begin{lemma}[$t_i$-validity]
${\GCone}$ satisfies sabotaged $t_i$-validity.
\end{lemma}
\begin{proof}
Suppose all honest parties propose the same value $v \in V$.
As there are at least $n - t_i$ honest parties and they all propose value $v$, every honest party eventually receives $n - t_i$ values for $v$ (and, by \cref{lemma:init_correct_proposal}, only $v$).
Therefore, every honest party decides $(v, 1)$, thus concluding the proof.
\qed
\end{proof}

\paragraph{Validity in the authenticated setting.}
We now prove that ${\GCone}$ satisfies validity in the authenticated setting.
Throughout the rest of the proof of the validity property, we suppose that all honest parties propose the same value $v$.

\begin{lemma} \label{lemma:validity_certificate}
If any $(n - t_i)$-certificate $\mathcal{C}(v')$ is formed, then $v' = v$.
\end{lemma}
\begin{proof}
As $\mathcal{C}(v')$ is formed and $n - t_i > t_s$ (given that $n > t_s + 2t_i$), there exists an honest party that sends an $\mathsf{init}(v')$ message for its proposal $v'$.
Given that all honest parties propose $v$, $v' = v$.
\qed
\end{proof}

\begin{lemma} \label{lemma:validity_round_2}
If any honest party locks some value $v'$ in round 2, then $v' = v$.
\end{lemma}
\begin{proof}
Let $p_i$ denote any honest party that locks some value $v'$ in round 2.
Therefore, $p_i$ receives an $\mathsf{init}(v')$ message for $v'$ from an honest party as $n - t_i > t_s$ (as $n > t_s + 2t_i$).
As all honest parties propose $v$, $v' = v$ and the lemma holds.
\qed
\end{proof}

\begin{lemma}[$t_s$-validity]
${\GCone}$ satisfies authenticated $t_s$-validity.
\end{lemma}
\begin{proof}
Consider any honest party $p_i$ that decides; let $p_i$ decide $(v', g')$.
We distinguish two cases:
\begin{compactitem}
    \item Let party $p_i$ lock in round $2$.
    In this case, $v' = v$ (due to \cref{lemma:validity_round_2}) and $g' = 1$.
    Therefore, the statement of the theorem holds in this case.

    \item Let party $p_i$ decide in round (i.e., step) $4$.
    Here, every honest party that proposes to $\GCPKI$ does so with value $v$.
    Indeed, if any honest party $p_j$ updates its $y_j$ variable in round $3$, \cref{lemma:validity_certificate} proves that $y_j$ holds value $v$.
    Hence, the validity property of $\GCPKI$ ensures that $v' = v$ and $g' = 1$.
\end{compactitem}
Having considered both cases, the proof is concluded.
\qed
\end{proof}

\paragraph{Consistency in the authenticated setting.}

\begin{lemma} \label{lemma:consistency_no_two_certificates}
No two $(n - t_i)$-certificates $\mathcal{C}(v)$ and $\mathcal{C}(v' \neq v)$ can be formed. 
\end{lemma}
\begin{proof}
By contradiction, suppose two certificates $\mathcal{C}(v)$ and $\mathcal{C}(v' \neq v)$ are formed.
Therefore, there are $n - t_i + n - t_i - n = n - 2t_i$ parties that participated in forming both certificates.
As $n > t_s + 2t_i$, $n - 2t_i > t_s$, which further implies that there is at least one honest party that participated in forming both certificates.
As this is impossible, we reach a contradiction.
\qed
\end{proof}

\begin{lemma} \label{lemma:consistency_round_2}
If any honest party locks $(v, 1)$ in round $2$, then all honest parties decide $(v, 1)$.
\end{lemma}
\begin{proof}
First, all honest parties that lock in round $2$ lock $(v, 1)$ (due to \cref{lemma:consistency_no_two_certificates}).
Moreover, every other honest party $p_i$ sets its $y_i$ variable to $v$ in round $3$ (due to \cref{lemma:consistency_no_two_certificates} and by protocol construction).
Therefore, the validity property of $\GCPKI$ ensures that all honest parties that decide in round $4$ decide $(v, 1)$.
\qed
\end{proof}

\begin{lemma}[$t_s$-consistency]
${\GCone}$ satisfies authenticated $t_s$-consistency.
\end{lemma}
\begin{proof}
To prove the lemma, we consider two possible scenarios:
\begin{itemize}
    \item There exist an honest party that locks $(v, 1)$ in round $2$.
    In this case, all honest parties decide $(v, 1)$ (by \cref{lemma:consistency_round_2}).

    \item No honest party decides in round $2$.
    In this case, the consistency property follows directly from the consistency property of $\GCPKI$.
\end{itemize}
As the statement of the lemma holds in both cases, the proof is concluded.
\qed
\end{proof}
\paragraph{Complexity.} The protocol $\GCPKI$ (using \cite{Momose2021}) has bit complexity of $O(\lambda n^2)$ for inputs of size $O(\lambda)$, and our addition to the protocol incurs an additive factor of $O(\lambda n^2)$ of bits communicated. Thus in total $\GCone$ has a bit complexity of $O(\lambda n^2)$. This concludes the proof of \cref{lemma:GC_PKI_star}.
\qed
\end{proof}

\subsection{Sabotaged Graded Consensus with Authenticated Validity}\label{ssec:GC_Sab_star}
We now move on to our second graded consensus protocol, the dual version of $\GCone$ for the sabotaged case. Specifically, we aim to design a protocol, $\GCtwo$, that functions as a standard GC in the sabotaged case, but maintains validity in the authenticated case. Formally, we prove the following.

\begin{theorem}\label{lemma:GCtwo}
    Let $t_s,t_i$ such that $2t_i+t_s<n, t_i\leq t_s<\frac{n}{2}$. There exists an $T_2=O(1)$ round deterministic protocol $\GCtwo$ that satisfies the following in the context of the graded consensus primitive (See \cref{def:GC}) in a $\Delta$-synchronous network if the round distance between protocol initiation times of any two honest parties is at most 1.
    \begin{itemize}
        \item Authenticated $t_s$-validity.
        \item Authenticated $t_s$-termination.
        \item Sabotaged $t_i$-security.
    \end{itemize}
Furthermore, given $\GCIT$ has at most $O(\lambda n^2)$ bit complexity and constant round complexity, so too does $\GCtwo$.
\end{theorem}
Our protocol $\GCtwo$ (\cref{fig:GC_IT_validity_in_PKI_world}) assumes a sabotaged $t_i$-secure graded consensus protocol, like the constant-round round protocol $\GCIT$ of Attiya and Welch \cite{AttiyaW23} that is sabotaged $t_i$-secure, and authenticated $t_s$-terminating for any $t_i\leq t_s < \frac{n}{2}, t_s+2t_i<n$, even in \emph{asynchrony}.
As in $\GCone$, parties first \emph{echo} their signed value.
As our goal is authenticated $t_s$-validity, parties form a $n - t_s$ certificate on $\echo(v)$ if possible.
In the third round, we now require that if parties have both received certificates only for one value $v$, and additionally received $n - t_s$ \emph{echos} of $v$, that they \emph{vote} for $v$ by multicasting $\vote(v)$.
In the fourth round, parties overwrite their input to $\GCIT$ if they receive at least $n - t_s - t_i$ $\vote(v)$ messages for unique $v$, lock their output to $(v, 1)$ if they receive $n - t_s$ $\vote(v)$ messages, and then execute $\GCIT$, outputting the result if there is no locked value.

\begin{figure}[ht!] 
\centering
\fbox{
\parbox{0.95\textwidth}{
\centering ${\GCtwo}(v_i)$ \\
(Pseudocode for a party $p_i$)\\
\textbf{Initialization: $y_i=v_i,g_i=0$}
\begin{itemize}

\item Round 1 (\textbf{Echo}): Multicast $\echo(v_i)$.
\item Round 2 (\textbf{Forward}): If $\exists v\in V$ s.t. received at least $n-t_s$ $\echo(v)$ messages, use TS to create certificate $\C(v)$ and  multicast it.
\item Round 3 (\textbf{Vote}): If $p$ made a certificate $\C(v)$ at Round 2 for value $v\in V$, didn't receive a certificate $\C(v')$ for any other value $v'\in V$, and received $n-t_s$ $\echo(v)$, multicast $\vote(v)$.
\item Round 4+: If $\exists v\in V$ s.t. received at least $n-t_s$ $\vote(v)$ messages, set $v_i=v$ and \textbf{lock} $y_i=v, g_i=1$. Otherwise if received at least $n-t_s-t_i$ $\vote(v)$ for a unique value, set $v_i=v$.  Run $(y_i, g_i) \leftarrow\GCIT(v_i)$.

\item \textbf{Decision:} Output $(y_i, g_i)$.

\end{itemize}

}
}
\caption{Sabotaged $t_i$-secure and authenticated $t_s$-valid $\GCtwo$ given a sabotaged $t_s$-secure graded consensus $\GCIT$.}
\label{fig:GC_IT_validity_in_PKI_world}
\end{figure}

\begin{proof}[of \cref{lemma:GCtwo}]
    Termination clearly holds in both settings by the behaviour of the protocol. We now move on to the other properties.

Let an honest party $p_i$ start executing $\GCtwo$ in some global round $\rho_i$.
Then, party $p_i$ executes Round $x \in \{ 1, 2, 3, 4 \}$ of $\GCtwo$ in global rounds $\rho_i + 2(x - 1)$ and $\rho_i + 2(x - 1) + 1$.

\paragraph{Authenticated $t_s$-validity.}
We first prove that $\GCtwo$ satisfies authenticated $t_s$-validity.

\begin{lemma}[$t_s$-consistency]
The protocol $\GCtwo$ satisfies authenticated $t_s$-validity.
\end{lemma}
\begin{proof}
Suppose all honest parties propose the same value $v$.
As all honest parties overlap in each round of $\GCtwo$ for (at least) $\delta$ time, all honest parties receive $n - t_s$ $\mathsf{echo}(v)$ messages at the start of Round 2.
Moreover, no certificate $\mathcal{C}(v' \neq v)$ can exist as that would imply that there exists a correct party whose proposal is $v' \neq v$.
Hence, every honest party multicast $\mathsf{vote}(v)$ in Round 3, which then implies that every honest party outputs $(v, 1)$.
Thus, the validity property is ensured in the signature world.
\qed
\end{proof}

\paragraph{Sabotaged $t_i$-validity.}
Then, we prove that $\GCtwo$ satisfies validity in the sabotaged setting.
We start by proving that if all honest parties propose the same value $v$ and an honest party creates a certificate $\mathcal{C}(v')$ in Round 2, then $v = v'$.

\begin{lemma} \label{lemma:certificate_round_2}
Suppose all honest parties propose the same value $v$.
If an honest party creates a certificate $\mathcal{C}(v')$ in Round 2, then $v = v'$.
\end{lemma}
\begin{proof}
If an honest party $p_i$ creates a certificate $\mathcal{C}(v')$ in Round 2, party $p_i$ has received an $\mathsf{echo}(v')$ message from an honest party as $n - t_s > t_i$ (given that $n > t_s + 2t_i$).
Therefore, $v' = v$.
\qed
\end{proof}
Next, we prove that if all honest parties propose the same value $v$ and an honest party multicast a $\mathsf{vote}(v')$ message in Round 3, then $v' = v$.

\begin{lemma} \label{lemma:vote_message_round_3}
Suppose all honest parties propose the same value $v$.
If an honest party multicasts a $\mathsf{vote}(v')$ message in Round 3, then $v' = v$.
\end{lemma}
\begin{proof}
If an honest party $p_i$ multicasts a $\mathsf{vote}(v')$ message in Round 3, party $p_i$ has previously constructed a certificate $\mathcal{C}(v')$ in Round 2.
By \Cref{lemma:certificate_round_2}, $v' = v$, thus concluding the proof.
\qed
\end{proof}
Next, we prove that if all honest parties propose $v$ to $\GCtwo$, then all honest parties propose $v$ to $\GCIT$ in Round 4.

\begin{lemma} \label{lemma:all_propose_v_round_4}
Suppose all honest parties propose the same value $v$.
Then, all honest parties propose $v$ to $\GCIT$ in Round 4.
\end{lemma}
\begin{proof}
As $v$ is proposed by every honest party, every honest party $p_i$ has $v_i = v$.
We show that if party $p_i$ updates its $v_i$ local variable, then it updates it to value $v$.
 Consider all possible places when party $p_i$ could update its local variable $v_i$:
\begin{itemize}
    \item Round 4 upon receiving $n - t_s$ $\mathsf{vote}(v')$ messages:
    As $n - t_s > t_i$ (given $n > t_s + 2t_i$), party $p_i$ receives a $\mathsf{vote}(v')$ message from an honest party.
    Therefore, \Cref{lemma:vote_message_round_3} proves that $v' = v$.
    Thus, the statement holds in this case.

    \item Round 4 upon receiving $n - t_s - t_i$ for a unique value $v'$:
    As $n - t_s - t_i > t_i$ (given $n > t_s + 2t_i$), party $p_i$ receives a $\mathsf{vote}(v')$ message from an honest party.
    Hence, $v' = v$ by \Cref{lemma:vote_message_round_3}, which proves the statement even in this case.
\end{itemize}
As $v_i$ remains $v$ at party $p_i$, the proof is concluded.
\qed
\end{proof}
Finally, we prove that $\GCtwo$ satisfies validity in the sabotaged setting.

\begin{lemma}[$t_i$-validity]
The protocol $\GCtwo$ satisfies validity in the sabotaged setting, i.e. $\GCtwo$ satisfies sabotaged $t_i$-validity.
\end{lemma}
\begin{proof}
Suppose all honest parties propose the same value $v$.
Consider any honest party $p_i$.
We distinguish two possible cases:
\begin{itemize}
    \item Let party $p_i$ decide a pair $(v', g')$ in Round 4.
    In this case, $p_i$ receives $n - t_s$ $\mathsf{vote}(v')$ messages.
    As $n - t_s > t_i$ (given $n > t_s + 2t_i$), $p_i$ receives a $\mathsf{vote}(v')$ message from an honest party.
    Therefore, \Cref{lemma:vote_message_round_3} proves that $v' = v$ and $g' = 1$.
    Thus, the validity property is satisfied in this case.

    \item Let party $p_i$ decide a value $v'$ after running the $\Pi_{\sfA\sfW}$ algorithm.
    By \Cref{lemma:all_propose_v_round_4}, all honest parties propose $v$ to $\GCIT$.
    Due to the validity property of $\GCIT$, every honest party decides $(v, 1)$ from $\GCIT$, thus concluding the proof even in this case.
\end{itemize}
As the validity property is satisfied in both cases, the proof is concluded.
\qed
\end{proof}

\paragraph{Consistency in the sabotaged setting.}
Next, we prove that $\GCtwo$ satisfies consistency in the sabotaged setting.
We first show that if one honest party sends a $\mathsf{vote}(v)$ message and another honest party sends a $\mathsf{vote}(v')$ message, then $v = v'$.

\begin{lemma} \label{lemma:one_vote}
If an honest party $p_i$ sends a $\mathsf{vote}(v)$ message and an honest party $p_j$ sends a $\mathsf{vote}(v')$ message, then $v = v'$.
\end{lemma}
\begin{proof}
By contradiction, let $v \neq v'$.
As $p_i$ (resp., $p_j$) sends a $\mathsf{vote}(v)$ (resp., $\mathsf{vote}(v')$) message, $p_i$ (resp., $p_j$) creates a certificate $\mathcal{C}(v)$ (resp., $\mathcal{C}(v')$) in Round 2.
Hence, party $p_i$ receives a certificate $\mathcal{C}(v')$ in Round 3 and party $p_j$ receives a certificate $\mathcal{C}(v)$ in Round 3.
Therefore, we reach contradiction with the fact that parties $p_i$ and $p_j$ send $\mathsf{vote}(\cdot)$ messages.
\qed
\end{proof}
Finally, we are ready to prove the consistency property of $\GCtwo$.

\begin{lemma}[$t_i$-consistency]
The protocol $\GCtwo$ satisfies consistency in the sabotaged setting.
\end{lemma}
\begin{proof}
We distinguish two cases:
\begin{compactitem}
    \item Let there exist an honest party $p_i$ that locks $(v, 1)$ in Round 4.
    Hence, honest party $p_i$ receives $n - t_s$ $\mathsf{vote}(v)$ messages in Round 4.
    Now, consider any honest party $p_j$.
    We further consider two scenarios:
    \begin{compactitem}
        \item Let $p_j$ lock $(v', g')$ in Round 4.
        In this case, $g' = 1$.
        Moreover, as $p_j$ decides in Round 4, $p_j$ receives $n - t_s - t_i$ $\mathsf{vote}(v')$ messages.
        As $n - t_s - t_i > t_i$, \Cref{lemma:one_vote} guarantees that $v' = v$.

        \item Let $p_j$ decide $(v', g')$ after deciding $(v', g')$ from $\GCIT$.
        As $p_i$ receives $n - t_s$ $\mathsf{vote}(v)$ messages in Round 4, every honest party receives (at least) $n - t_s - t_i$ $\mathsf{vote}(v)$ messages in Round 4.
        Importantly, as $n - t_s - t_i > t_i$, \Cref{lemma:one_vote} guarantees that no honest party receives $n - t_s - t_i$ $\mathsf{vote}(v')$, for any value $v' \neq v$.
        Hence, every honest party $p_i$ sets its local variable $v_i$ to $v$ and proposes $v_i$ to $\GCIT$.
        Finally, the strong validity property of $\GCIT$ ensures $v' = v$.
    \end{compactitem}
    
    \item  No honest party locks $(v, 1)$ in Round 4.
    In this case, the consistency property follows directly from the consistency property of $\GCIT$.
\end{compactitem}
The lemma holds.
\qed
\end{proof}

\paragraph{Complexity.} The protocol $\GCIT$ (using \cite{AttiyaW23}) has bit complexity of $O(\lambda n^2)$ for inputs of size $O(\lambda)$, and our addition to the protocol incurs an additive factor of $O(\lambda n^2)$ of communication bits. Thus in total $\GCtwo$ has bit complexity of $O(\lambda n^2)$. This concludes the proof of \cref{lemma:GCtwo}.
\qed
\end{proof}
\section{Conclusion}\label{sec:conclusion}
In this work, we have constructed efficient crypto-agnostic Byzantine agreement, and in particular a protocol with $O(\lambda n^2)$ bit complexity and constant round complexity in the authenticated setting.
Natural open problems are as follows:
\begin{itemize}
\item Our \name\ protocols use $O(\lambda n^2)$ bits only when the input message is of size $O(\lambda)$, and otherwise $O(L n^2)$ for $L = \Omega(\lambda)$. It is thus natural to consider efficient crypto-agnostic BA for \emph{long messages}.
The main difficulty here is keeping complexity low while also providing security in the sabotaged or information-theoretic setting where it is difficult enough to build efficient protocols~\cite{civit2024error} let alone in the crypto-agnostic setting.
\item As \name\ optimises round complexity in the authenticated case, it may be of interest to instead optimise for the \emph{sabotaged} case (i.e., not running the authenticated protocol in `good' executions).
\item Finally, extending our results to the model where the public key infrastructure may be inconsistent or arbitrarily broken as considered in~\cite{EC:FitHolMul04,cryptoeprint:2009/434} may be of interest.
\end{itemize}

\iffull
\section*{Acknowledgements}
Daniel Collins was supported in part by AnalytiXIN and by Sunday Group, Inc.
\fi

\bibliographystyle{alpha}
{
\bibliography{reportbib,references,
	abbrev3,crypto}
}

\newpage
\appendix

\end{document}